\documentclass[11pt]{article}

\usepackage[sort]{cite}       
\usepackage{amssymb}          
\usepackage{latexsym}         
\usepackage{amsmath}          
\usepackage{amscd}            
\usepackage[latin1]{inputenc} 
\usepackage{eucal}            
\usepackage{bbm}              
\usepackage{theorem}          
\usepackage{stmaryrd}         
\usepackage{mathrsfs}         
\usepackage{paralist}         
\usepackage[all]{xy}

\title{Complete Positivity of Rieffel's Deformation Quantization by
  Actions of $\mathbb{R}^d$}

\author{\textbf{Daniel Kaschek}\thanks{Daniel.Kaschek@physik.uni-freiburg.de},
  \addtocounter{footnote}{1}
  \textbf{Nikolai Neumaier}\thanks{Nikolai.Neumaier@physik.uni-freiburg.de},
  \textbf{Stefan Waldmann}\thanks{Stefan.Waldmann@physik.uni-freiburg.de},
  \\[0.1cm]
  Fakult{\"a}t f{\"u}r Mathematik und Physik\\
  Albert-Ludwigs-Universit{\"a}t Freiburg\\
  Physikalisches Institut\\
  Hermann Herder Stra{\ss}e 3\\
  D 79104 Freiburg\\
  Germany
  }

\date{September 2008}

\renewcommand{\mathbb}[1]{\mathbbm{#1}} 

\numberwithin{equation}{section}

\theoremheaderfont{\normalfont\bfseries}
\theorembodyfont{\itshape}
\newtheorem{lemma}{Lemma}[section]
\newtheorem{proposition}[lemma]{Proposition}
\newtheorem{theorem}[lemma]{Theorem}
\newtheorem{corollary}[lemma]{Corollary}

\theorembodyfont{\rmfamily}

\newtheorem{definition}[lemma]{Definition}

\newtheorem{remark}[lemma]{Remark}

\newenvironment{proof}[1][{}]{\par\noindent\textsc{Proof{#1}: }}{\hspace*{\fill}$\blacksquare$\smallskip\noindent\par}

\newcommand{\cc}[1]      {\overline{{#1}}}              
\newcommand{\id}         {\operatorname{\mathsf{id}}}   
\newcommand{\SP}[1]      {\left\langle{#1}\right\rangle} 
\newcommand{\Unit}       {\mathbb{1}}                  
\newcommand{\norm}[1]    {\left\|{#1}\right\|}            
\newcommand{\I}          {\mathrm{i}}
\newcommand{\E}          {\mathrm{e}}
\newcommand{\D}          {\operatorname{\mathrm{d}}}

\begin{document}

\maketitle

\begin{abstract}
    In this paper we consider $C^*$-algebraic deformations by actions
    of $\mathbb{R}^d$ \`{a} la Rieffel and show that every state of
    the undeformed algebra can be deformed into a state of the
    deformed algebra in the sense of a continuous field of states. The
    construction is explicit and involves a convolution operator with
    a particular Gauß function.
\end{abstract}

\noindent
\textbf{Keywords:} Rieffel deformation quantization. Completely
positive deformation. Continuous fields.

\noindent
\textbf{MSC (2008):} 53D55, 46L87, 81R60, 46L65.

\newpage

\tableofcontents

%
%

\section{Introduction}
\label{sec:Intro}

In deformation quantization \cite{bayen.et.al:1978a} the transition
from classical mechanics to quantum mechanics is obtained as an
associative deformation of the classical observable algebra, modelled
by a certain class of functions on the classical phase space. In
\emph{formal} deformation quantization this is accomplished by
constructing a new associative product, the star product, as a formal
power series with formal parameter $\hbar$. While this theory is by now very
well understood, see \cite{kontsevich:2003a, dewilde.lecomte:1983b,
  fedosov:1996a, nest.tsygan:1995a, deligne:1995a,
  bertelson.cahen.gutt:1997a} for existence and classification
results and \cite{waldmann:2007a} for a gentle introduction, from a
physicist's perspective the formal character of the star products is still
not satisfying: $\hbar$ is not a formal parameter after all, whence at
the end of the day, some sort of ``convergence'' in $\hbar$ is
needed.

Attacking the convergence problem of the formal series seems to be
complicated though in examples this can be done
\cite{beiser.roemer.waldmann:2007a}. More successful are approaches
which are intrinsically non-formal like the Berezin-Toeplitz inspired
quantizations \cite{cahen.gutt.rawnsley:1995a,
  cahen.gutt.rawnsley:1994a, cahen.gutt.rawnsley:1993a,
  cahen.gutt.rawnsley:1990a,
  bordemann.meinrenken.schlichenmaier:1991a} or Rieffel's approach
using oscillatory integrals based on group actions of $\mathbb{R}^d$.
In this version \cite{rieffel:1993a}, the starting point is a
$C^*$-algebra $\mathfrak{A}$ endowed with an isometric, strongly
continuous action by $^*$-automorphisms by some finite-dimensional
vector space $V$. Out of this and the choice of a symplectic form on
$V$, Rieffel constructs a deformation of $\mathfrak{A}$ in the sense
of a continuous field of $C^*$-algebras, the field parameter being
$\hbar$. While the construction is very general, there are yet many
examples of Poisson manifolds which can be deformation quantized this
way. In this framework of strict deformations many results have been
obtained, most notably \cite{natsume.nest.peter:2003a,
  landsman:1998a}.

While the above constructions deal with the observable algebra, for a
physically complete description of quantization also the states have
to be taken into account. In both approaches, the appropriate notion
of states is that of positive linear functionals on the observable
algebras. While for $C^*$-algebras this is of course a well-known
concept, also in the formal deformation quantization this leads to a
physically reasonable definition incorporating a reasonable
representation theory, see e.g.\ \cite{bordemann.waldmann:1998a,
  bursztyn.waldmann:2005b, waldmann:2005b} and references therein.

A fundamental question is whether a given classical state arises as
the classical limit of a quantum state. In formal deformation
quantization there is a general and affirmative answer to this
question \cite{bursztyn.waldmann:2005a, bursztyn.waldmann:2000a}. In
the strict approaches, Landsman discussed this in
\cite{landsman:1993b} for a certain class of examples: the appropriate
notion of classical limit and deformation of states is that of a
continuous field of states with respect to a given continuous field of
$C^*$-algebras. His construction is based on particular
$^*$-representations and certain coherent states and their Wigner
functions. More recently, Landsman uses continuous fields of states in
his discussion of the Born rule \cite{landsman:2008a}.

In this paper, we consider Rieffel's deformation by actions of
$\mathbb{R}^d$ in general and prove that every state of the undeformed
algebra can be deformed into a continuous field of states for the
field of deformed algebras.  Moreover, we give an explicit
construction including a detailed study of the asymptotics of the
deformed states for $\hbar \longrightarrow 0$, see also
\cite{kaschek:2008a}. It turns out that the asymptotic expansion
coincides in a very precise sense with the formal deformations
obtained in \cite{bursztyn.waldmann:2000a}.

The paper is organized as follows: in Section~\ref{sec:SonAinfty} we
recall Rieffel's deformation in the Fr\'echet algebraic framework and
define an operator $S_\hbar$ being the ``convolution'' with a Gauß
function. The precise form of $S_\hbar$ resembles the Wigner
functions Landsman used, however now $S_\hbar$ is defined directly on
the algebra.  The asymptotics of $S_\hbar$ for $\hbar \longrightarrow
0^+$ is studied in detail.  In Section~\ref{sec:DeformationPosFun} we
show that $S_\hbar$ maps squares $a^* \star_\hbar a$ of the deformed
algebra to positive elements of the undeformed algebra. This allows to
define a positive functional $\omega_\hbar = \omega \circ S_\hbar$ of
the deformed algebra for every positive functional $\omega$ of the
undeformed algebra. A detailed asymptotic expansion is obtained as well.
Section~\ref{sec:CstarCase} is devoted to the more particular case of
a $C^*$-algebra deformation. Here we show that the operator $S_\hbar$
is also continuous in the $C^*$-topology of the deformed algebra
whence it extends to the $C^*$-algebraic completion.  Finally, in
Section~\ref{sec:FieldOfStates} we show that the positive functionals
$\{\omega_\hbar\}_{\hbar \ge 0}$ indeed form a continuous field of
states.

\noindent
\textbf{Acknowledgements:} We would like thank Klaas Landsman for
pointing out reference \cite{landsman:2008a} as well as Marc Rieffel
and the referee for valuable comments.

%
%

\section{The Operator $S$ on $\mathcal A^{\infty}$}
\label{sec:SonAinfty}

In this section $\mathcal{A}$ denotes a Fr\'echet $^*$-algebra endowed
with a strongly continuous action $\alpha$ by $^*$-homomorphisms of a
finite-dimensional vector space $V$ which we assume without
restriction to be even dimensional. Moreover, one requires that there
is a system of seminorms $\norm{\cdot}_k$ defining the topology of
$\mathcal{A}$ such that with respect to these seminorms the action is
isometric.  By $\mathcal{A}^\infty \subseteq \mathcal{A}$ we denote
the subspace of smooth vectors in $\mathcal{A}$ with respect to
$\alpha$.  It is well known that $\mathcal{A}^\infty$ is a dense
$^*$-subalgebra of $\mathcal{A}$. Moreover, $\mathcal{A}^\infty$
carries a finer topology making it into a Fr\'echet algebra, too.  A
system of seminorms defining the topology is explicitly given by
\begin{equation}
    \label{eq:SeminormsSmoothVectors}
    \norm{a}_{k,\mu}
    = \sup_{|\beta|\leq \mu}\norm{\partial^{\beta} a}_{k},
\end{equation}
where using multi-index notation $\partial^\beta a$ denotes the derivative of
$\alpha_u(a)$ with respect to $u$ at $u = 0$, see e.g.\ \cite{taylor:1986a} for more
background on smooth vectors.

In a next step one chooses a non-degenerate bilinear anti-symmetric
form $\theta$ on $V$ and $\hbar > 0$. Then Rieffel showed in
\cite{rieffel:1993a} that
\begin{equation}
    \label{eq:TheStarProduct}
    a \star_{\hbar} b 
    =
    \frac{1}{(\pi\hbar)^{2n}}
    \int_{V \times V}
    \alpha_u(a) \alpha_v(b) \E^{\frac{2\I}{\hbar}\theta(u,v)} \D(u,v),
\end{equation}
which is defined for $a, b \in \mathcal{A}^\infty$, yields a
well-defined associative product such that $\star_\hbar$ is still
continuous with respect to the $\mathcal{A}^\infty$-topology.
Moreover, the original $^*$-involution of $\mathcal{A}^\infty$ is
still a $^*$-involution with respect to $\star_\hbar$. The precise
definition of the integral in an oscillatory sense is sophisticated
and can be found in Rieffel's booklet \cite{rieffel:1993a}. Note that
we have to choose a normalization for the Haar measure on $V$ in
\eqref{eq:TheStarProduct}. We shall also make use of linear
coordinates denoted by $v = u_i e_i$ in the sequel.

\begin{definition}[The Operator $S_g$]
    \label{definition:S}
    Let $g: V\times V \longrightarrow \mathbb{R}$ be a positive
    definite inner product on $V$. Then the linear operator $S_g:
    \mathcal{A} \longrightarrow \mathcal{A}$ is defined by
    \begin{equation}
        \label{eq:SDef}
        S_g(a) = \int_V \E^{-g(u,u)}\alpha_u(a) \D u.
    \end{equation}
\end{definition}
Thanks to the fast decay of the Gauß function and the fact that the
action $\alpha$ is isometric, the definition of $S_g$ in
\eqref{eq:SDef} as an improper Riemann integral is possible.  More
general, we need the following construction: let $B(V, \mathcal{A})$
denote the $\mathcal{A}$-valued functions on $V$ such that $\sup_{v
  \in V}\norm{f(v)}_k < \infty$ for all $k$, i.e.\ the bounded
functions with respect to the seminorms $\norm{\cdot}_k$ of
$\mathcal{A}$. Then we define
\begin{equation}
    \label{eq:TildeSgDef}
    \tilde{S}_g f = \int_V \E^{-g(u,u)} f(u) \D u
\end{equation}
for $f \in B(V, \mathcal{A})$. Again, a naive definition of the
integral is possible. Finally, let $C^0_u(V, \mathcal{A})$ be the
uniformly continuous functions in $B(V, \mathcal{A})$ and let
$C^\infty_u(V, \mathcal{A})$ be the smooth functions with all partial
derivatives in $C^0_u(V, \mathcal{A})$. Clearly, the spaces $C^0_u (V,
\mathcal{A})$ as well as $C^\infty_u(V, \mathcal{A})$ are equipped
with a natural Fr\'echet topology by taking the sup-norm over $V$ of
seminorms of the values of the (derivatives of the) functions.  Then
the following Proposition lists some properties of $S_g$ and
$\tilde{S}_g$:
\begin{proposition}[Continuity of $S_g$]
    \label{proposition:ContinuityOfS}
    ~
    \begin{enumerate}
    \item The operator $S_g: \mathcal{A} \longrightarrow \mathcal{A}$
        is continuous.
    \item We have $S_g(\mathcal{A}^\infty) \subseteq
        \mathcal{A}^\infty$ and $S_g: \mathcal{A}^\infty
        \longrightarrow \mathcal{A}^\infty$ is continuous, too.
    \item The restriction of $\tilde{S}_g$ to $C^0_u (V,
        \mathcal{A})$ and $C^\infty_u(V, \mathcal{A})$
        is continuous in the respective topologies.
    \item The restriction of $\tilde{S}_g$ to $C^0_u(V,
        \mathcal{A}^\infty)$ and $C^\infty_u(V, \mathcal{A}^\infty)$
        takes values in $\mathcal{A}^\infty$ and is again continuous.
    \end{enumerate}
\end{proposition}
\begin{proof}
    The first two statements can be recovered from the third and
    fourth by considering the function $f(u) = \alpha_u(a)$ for $a \in
    \mathcal{A}$ or $a \in \mathcal{A}^\infty$, respectively: as the
    action is isometric we have $f \in C^0_u(V, \mathcal{A})$ and
    $C^\infty_u(V, \mathcal{A}^\infty)$, respectively. The continuity
    statements in the third and fourth part are then a straightforward
    estimate.
\end{proof}

In a next step we want to understand the asymptotics of the operator
$S_g$. To this end we rescale the inner product by $\hbar > 0$ and
consider the normalized Gauß function
\begin{equation}
    \label{eq:NormalizedGauss}
    G_\hbar(u)
    = \frac{\sqrt{\det{G}}}{(\pi\hbar)^{n}} \E^{-\frac{g(u,u)}{\hbar}}
\end{equation}
where $\det{G} > 0$ is the determinant of $g$ with respect to the Haar
measure on $V$ and $2n = \dim V$. The normalization constant is chosen
such that the integral of $G_\hbar$ is $1$. For a fixed choice of
$g$ we consider the operator
\begin{equation}
    \label{eq:ShbarDef}
    S_\hbar (a) = \int_V G_\hbar(u) \alpha_u(a) \D u.
\end{equation}
\begin{lemma}
    \label{lemma:asymp}
    For every $a \in \mathcal A$ we have $\lim_{\hbar \searrow 0}
    S_{\hbar}(a) = a$ in the topology of $\mathcal{A}$.  Moreover, for
    every $a \in \mathcal{A}^\infty$ we have
    \begin{equation}
        \label{eq:ShbarToIdAinftyTop}
        \lim_{\hbar \searrow 0} S_{\hbar}(a) = a
    \end{equation}
    and
    \begin{equation}
        \frac{\D}{\D\hbar}
        \bigl(S_{\hbar}a\bigr) 
        = S_{\hbar}\left(\frac{1}{4}\Delta_g a\right),
        \label{eq:iteration1}
    \end{equation}
    both with respect to the topology of $\mathcal A^{\infty}$ where
    \begin{equation}
        \label{eq:DefinitionLaplacian}
        \Delta_g a = 
        \sum_{i, j} (G^{-1})^{ij}
        \left.
            \frac{\partial^2}{\partial u_i \partial u_j} \alpha_u(a)
        \right|_{u=0}
    \end{equation}
    is the Laplacian with respect to the inner product $g$ and the
    action $\alpha$ viewed as continuous operator on
    $\mathcal{A}^\infty$. The operator $\Delta_g$ does not depend on
    the choice of linear coordinates.
\end{lemma}
\begin{proof}
    By substitution $u \longrightarrow \sqrt{\hbar} u$ we have
    \[
    S_{\hbar}(a) 
    =
    \frac{\sqrt{\det{G}}}{\pi^n} \int_V  
    \E^{-g(u,u)} \alpha_{\sqrt{\hbar}u}(a) \D u.
    \]
    To exchange the order of integration and $\lim\limits_{\hbar
      \searrow 0}$ we consider
    \[
    \begin{split}
        &\norm{
          \int_V \E^{-g(u,u)}(\alpha_{\sqrt{\hbar}u}(a)-a) \D u 
        }_{k,\mu}\\
        &\quad\leq
        \int_K \E^{-g(u,u)} 
        \norm{\alpha_{\sqrt{\hbar}u}(a)-a}_{k,\mu} \D u 
        +
        \int_{V\setminus K} \E^{-g(u,u)}
        \norm{\alpha_{\sqrt{\hbar}u}(a)-a}_{k,\mu} \D u,
    \end{split}
    \]
    where $K$ denotes a compact set in $V$. For $\hbar \searrow 0$
    the function $\alpha_{\hbar}(a): u \mapsto
    \alpha_{\sqrt{\hbar}u}(a)$ converges uniformly to the constant
    function $u \mapsto a$ on every compact set in $V$. Furthermore,
    since $\alpha$ is isometric, the estimate
    $\norm{\alpha_{\sqrt{\hbar}u}(a) - a}_{k, \mu} \leq 2\norm{a}_{k,
      \mu}$ holds for all $u\in V$. Thus, choosing $K$ large enough makes
    the second term small, independently of $\hbar$.  Afterwards,
    choosing $\hbar$ small brings the first term for the fixed $K$
    below every positive bound.  By the choice of the normalization
    constant in front of the Gauß function this shows
    \eqref{eq:ShbarToIdAinftyTop}. The case for $a \in \mathcal{A}$ is
    analogous. For the last statement we first note that for a fixed
    $a$ the differentiation in $V$-directions is a limit in
    $C^\infty_u(V, \mathcal{A}^\infty)$. By the linearity and
    continuity of $\tilde{S}_g$ as in
    Proposition~\ref{proposition:ContinuityOfS} we can thus exchange
    differentiation and the integral. This gives
    \begin{align*}
        \frac{\D}{\D\hbar} S_{\hbar} a 
        &= 
        \frac{\D}{\D\hbar}
        \int_V \frac{\sqrt{\det{G}}}{\pi^n} 
        \E^{-g(u,u)} 
        \alpha_{\sqrt{\hbar}u}(a) \D u \\
        &=
        \frac{\sqrt{\det{G}}}{\pi^n} \int_V
        \E^{-g(u,u)}
        \sum_i \frac{u_i}{2\sqrt{\hbar}}
        \left.
            \frac{\partial}{\partial u_i} \alpha_{u}(a)
        \right|_{\sqrt{\hbar} u}
        \D u \\ 
        &=
        - \frac{1}{4 \sqrt{\hbar}} \frac{\sqrt{\det{G}}}{\pi^n} 
        \int_V 
        \sum_{i, j} (G^{-1})^{ij} \frac{\partial}{\partial u_j} \E^{-g(u,u)}
        \left.
            \frac{\partial}{\partial u_i} \alpha_{u}(a)
        \right|_{\sqrt{\hbar} u}
        \D u \\ 
        &=
        \frac{1}{4} \frac{\sqrt{\det{G}}}{\pi^n}
        \E^{-g(u,u)}
        \left.
            (G^{-1})^{ij} \frac{\partial^2}{\partial u_i \partial u_j}
            \alpha_u(a)
        \right|_{\sqrt{\hbar}u}
        \D u \\
        &=
        \frac{1}{4} \frac{\sqrt{\det{G}}}{\pi^n}
        \E^{-g(u,u)}
        \alpha_{\sqrt{\hbar}u}
        \left(
            \left.
                (G^{-1})^{ij}
                \frac{\partial^2}{\partial v_i \partial v_j}
                \alpha_v(a)
            \right|_{v = 0}
        \right)
        \D u,
    \end{align*}
    where we have used an integration by parts as well as the fact
    that $\alpha$ is an action. Note that the operator $\Delta_g$ is
    well-defined on $\mathcal{A}^\infty$. This completes the proof.
\end{proof}

Since with $a \in \mathcal{A}^\infty$ we also have $\Delta_g a \in
\mathcal{A}^\infty$, the iteration of \eqref{eq:iteration1}
immediately yields the following statement:
\begin{theorem}[Asymptotic expansion of $S_{\hbar}$]
    \label{theorem:AsymptoticOfShbar}
    The operator $S_\hbar:\mathcal{A}^\infty$ $\longrightarrow
    \mathcal{A}^\infty$ has the formal asymptotic expansion
    \begin{equation}
        \label{eq:AsymptoticShbar}
        S_{\hbar}
        \simeq_{\hbar\searrow 0}
        \E^{\frac{\hbar}{4}\Delta_g}
    \end{equation}
    with respect to the topology of $\mathcal{A}^\infty$.
\end{theorem}
This means that the asymptotic expansion of $S_{\hbar}$ corresponds to
the formal equivalence transformation leading from the Weyl star
product to the Wick product, see
e.g.~\cite[Eq.~(5.84)]{waldmann:2007a}.

%
%

\section{Deformation of Positive Functionals}
\label{sec:DeformationPosFun}

Recall that a functional $\omega: \mathcal{A} \longrightarrow
\mathbb{C}$ is called positive if for all $a \in \mathcal{A}$ we have
\begin{equation}
    \label{eq:PositiveFunctional}
    \omega(a^*a) \ge 0.
\end{equation}
While this is a purely algebraic definition, for a topological algebra
$\mathcal{A}$ we require furthermore that $\omega$ is
\emph{continuous}. An algebra element $a \in \mathcal{A}$ is called
\emph{positive} if $\omega(a) \ge 0$ for all (continuous) positive
functionals $\omega$. The positive algebra elements will be denoted by
$\mathcal{A}^+$. Note that for general $^*$-algebras a definition of
positivity like $a = b*b$ will not lead to a reasonable notion of
positive elements due to the lack of a functional calculus. Note also
that the above definition coincides with the usual definition of
positive elements in case of a $C^*$-algebra.  There are even more
sophisticated notions of positivity, e.g.\ for $O^*$-algebras, see the
discussion in \cite{schmuedgen:1990a}. However, for our purposes the
above definition will be sufficient as for $C^*$-algebras positive
functionals are always continuous.

We can now use the operator $S_\hbar$ to deform a positive functional
of $\mathcal{A}$ into a positive functional with respect to
$\star_\hbar$. To this end we observe the following lemma:
\begin{lemma}
    \label{lemma:Shbarastara}
    For $a \in \mathcal{A}^\infty$ we have
    \begin{equation}
        \label{eq:Shbarastara}
        \begin{split}
            &S_\hbar (a^* \star_\hbar a) \\
            &\quad=
            \frac{1}{(\pi\hbar)^{2n}} \int_{V \times V}
            \E^{-\frac{1}{\hbar}g(v,v)}
            \alpha_v(a^*)
            \E^{-\frac{1}{\hbar}g(w,w)}
            \alpha_w(a)
            \E^{\frac{2}{\hbar}(g(v,w) + \I\theta(v,w))} \D v \D w.
        \end{split}
    \end{equation}
\end{lemma}
\begin{proof}
    The proof is a straightforward computation using the fact that
    $\alpha$ is an action as well as a linear change of coordinates
    and a Fourier transform of the Gauß function.
\end{proof}

In the particular case that $g$ and $\theta$ are \emph{compatible},
i.e.\ $g(u,v) = \theta(u, Jv)$ with a complex structure $J$, the
combination $h(u,v) = g(u,v) + \I\theta(u,v)$ is known to be a
Hermitian metric on the complex vector space $(V,J)$.  In this
particular case there exists a symplectic basis
$\{e_1,\dots,e_n,f_1,\dots,f_n\}$ of $V$ with coordinates $q^i$ and
$p^i$ and there exist complex coordinates $z^i = q^i + \I p^i$ and
$\cc{z}^i = q^i - \I p^i$ such that
\begin{equation}
    \label{eq:ghInNiceCoords}
    g(u, u) = \sum\limits_i z_u^i \cc{z}_u^i = \norm{z_u}^2
    \quad
    \textrm{and}
    \quad
    h(v,w) = \sum\limits_i \cc{z}_v^i z_w^i.
\end{equation}
From now on we assume that $g$ is compatible with $\theta$.  Using
these coordinates, the above integral can be rewritten as
\begin{equation}
    \label{eq:ShbarastaraComplex}
    S_\hbar(a^* \star_\hbar a)
    =
    \frac{1}{(\pi\hbar)^{2n}}
    \int_{V \times V}
    \E^{-\frac{1}{\hbar}\norm{z_v}^2} 
    \alpha_v (a^*)
    \E^{-\frac{1}{\hbar}\norm{z_w}^2}
    \alpha_w(a)
    \E^{\frac{2}{\hbar}\bar z_v\cdot z_w}
    \D v \D w.
\end{equation}
\begin{lemma}
    \label{lemma:ExchangeExpSeriesIntegral}
    For $a \in \mathcal{A}^\infty$ we have
    \begin{equation}
        \label{eq:ShbarSeries}
        S_\hbar (a^* \star_\hbar a)
        =
        \sum_{L \ge 0} \frac{2^{|L|}}{L!} a_L^* a^{\phantom{*}}_L
    \end{equation}
    with respect to the $\mathcal{A}^\infty$-topology where for a
    multi-index $L = (l_1, \ldots, l_n)$ one defines
    \begin{equation}
        \label{eq:aKdef}
        a^{\phantom{*}}_L = 
        \frac{1}{\pi^n} \int_V
        \E^{-\norm{z_v}^2}
        z_v^L \alpha_{\sqrt{\hbar} v}(a) \D v.
    \end{equation}
\end{lemma}
\begin{proof}
    First note that rescaling the variables in
    \eqref{eq:ShbarastaraComplex} by $\sqrt{\hbar}$ allows to get rid
    of the negative powers of $\hbar$. Then \eqref{eq:ShbarSeries} is
    obtained from expanding the exponential function $\E^{2 \cc{z}_v
      \cdot z_w}$ into the Taylor series and exchanging summation and
    integration. The fact that the latter exchange of limits is allowed
    follows from a similar argument as in the proof of
    Lemma~\ref{lemma:asymp}: First we split the integration into two
    parts, one over a compact subset $K \subseteq V$ and the other
    over $V \setminus K$. On $K$ the Taylor expansion converges
    uniformly including all derivatives. Outside $K$, the Gauß
    function decays fast enough to over-compensate the exponential
    increase. Thus first choosing $K$ large enough to make the second
    integral small then using the uniform convergence gives the
    result. Note that the convergence is in the sense of
    $\mathcal{A}^\infty$.
\end{proof}
\begin{theorem}[Positive Deformation of $\omega$]
    \label{theorem:positiv}    
    Let $g$ be a compatible positive definite inner product and
    $S_\hbar$ the corresponding operator as in \eqref{eq:ShbarDef}.
    \begin{enumerate}
    \item For every continuous positive linear functional $\omega:
        \mathcal{A} \longrightarrow \mathbb{C}$ the functional
        \begin{equation}
            \label{eq:omegahbarDef}
            \omega_\hbar = \omega \circ S_\hbar:
            \mathcal{A}^\infty \longrightarrow \mathbb{C}
        \end{equation}
        is positive and continuous in the
        $\mathcal{A}^\infty$-topology.
    \item For every $a \in \mathcal{A}^\infty$ we have
        \begin{equation}
            S_{\hbar}(a^* \star_\hbar a) \in \mathcal{A}^+.
        \end{equation}
    \end{enumerate}
\end{theorem}
\begin{proof}
    Let $\omega: \mathcal{A} \longrightarrow \mathbb{C}$ be positive
    and continuous. Since the topology of $\mathcal{A}^\infty$ is finer
    than the original one, $\omega: \mathcal{A}^\infty \longrightarrow
    \mathbb{C}$ is still continuous. Then $\omega(S_\hbar(a^*
    \star_\hbar a)) \ge 0$ follows immediately from
    \eqref{eq:ShbarSeries} and the continuity of $\omega$. Moreover,
    since $S_\hbar$ is continuous the first part follows. Thus the
    second part is clear.
\end{proof}
\begin{corollary}
    \label{corollary:AsymptoticOfomegahbar}
    Let $\omega: \mathcal{A} \longrightarrow \mathbb{C}$ be a positive
    and continuous linear functional. Then on $\mathcal{A}^\infty$,
    $\omega_\hbar = \omega \circ S_\hbar$ has the asymptotic expansion
    \begin{equation}
        \label{eq:omegaAsymptotic}
        \omega_\hbar 
        \simeq_{\hbar\searrow 0}
        \sum_{r=0}^\infty \frac{1}{r!} \left(\frac{\hbar}{4}\right)^r 
        \omega \circ \Delta_g^r       
    \end{equation}
    in the $\mathcal{A}^\infty$-topology.
\end{corollary}
\begin{remark}
    \label{remark:AsinFormalDQ}
    This kind of formal positive deformation of a positive functional
    was discussed in \cite{bursztyn.waldmann:2000a} based on the
    formal equivalence between the Weyl and Wick star product.
\end{remark}

%
%

\section{The Operator $S$ in the $C^*$-Case}
\label{sec:CstarCase}

In a next step we want to apply Theorem~\ref{theorem:positiv} to the
more particular case of a $C^*$-algebraic deformation. Let
$\mathfrak{A}$ be a unital $C^*$-algebra endowed with an isometric and
strongly continuous action of $V$ by $^*$-automorphisms. Then Rieffel
has shown how to construct a $C^*$-norm on the Fr\'echet $^*$-algebra
$\mathcal{A}(\hbar) = (\mathfrak{A}^\infty, \star_\hbar, ^*)$.  In
general, $\mathcal{A}(\hbar)$ is not complete. The norm completion of
$\mathcal{A}(\hbar)$ will then be denoted by $\mathfrak{A}(\hbar)$.

We briefly recall the construction of the $C^*$-norm on
$\mathcal{A}(\hbar)$. Let $\mathcal{S}(V, \mathfrak{A})\subseteq
C^\infty_u(V, \mathfrak{A})$ be the subset of functions which are
still in $C^\infty_u(V, \mathfrak{A})$ when multiplied by arbitrary
polynomials on $V$. For $f, g \in \mathcal{S}(V, \mathfrak{A})$ one
defines the $\mathfrak{A}$-valued inner product
\begin{equation}
    \label{eq:AvaluedInnerProduct}
    \SP{f, g} = \int_V f(v)^* g(v) \D v,
\end{equation}
which makes $\mathcal{S}(V, \mathfrak{A})$ into a pre-Hilbert right
$\mathfrak{A}$-module, see e.g.\ \cite{lance:1995a} for details on
Hilbert modules. In particular, by
\begin{equation}
    \label{eq:HilbertModuleNorm}
    \norm{f}_{\mathcal{S}} = \sqrt{\norm{\SP{f, f}}}
\end{equation}
one obtains a norm on $\mathcal{S}(V, \mathfrak{A})$, where the norm
on the right hand side is the $C^*$-norm of $\mathfrak{A}$. Using this
norm, Rieffel showed that for every $F \in C^\infty_u(V,
\mathfrak{A})$ the operator
\begin{equation}
    \label{eq:FstarOperator}
    F \star_\hbar \cdot : \mathcal{S}(V, \mathfrak{A}) 
    \ni f \; \mapsto \; F \star_\hbar f \in
    \mathcal{S}(V, \mathfrak{A})
\end{equation}
is continuous with respect to $\norm{\cdot}_{\mathcal{S}}$ and
adjointable with adjoint given by $F^* \star_\hbar \cdot$. Since for
$a \in \mathfrak{A}^\infty$ the function $u \mapsto \alpha_u(a)$ is in
$C^\infty_u(V, \mathfrak{A})$ we obtain an induced operator on the
pre-Hilbert module $\alpha(a) \star_\hbar \cdot$ which is continuous
and adjointable. A final computation then shows that $a \mapsto
\alpha(a) \star_\hbar \cdot$ is a $^*$-homomorphism with respect to
the deformed product $\star_\hbar$ of $\mathfrak{A}^\infty$. This
allows to define 
\begin{equation}
    \label{eq:CstarNormhbar}
    \norm{a}_\hbar = \norm{\alpha(a) \star_\hbar \cdot},
\end{equation}
where on the right hand side we use the operator norm. Since it is
well known that the continuous and adjointable operators on a
(pre-)Hilbert module constitute a $C^*$-algebra, Rieffel arrives at a
$C^*$-norm $\norm{\cdot}_\hbar$ for $\mathcal{A}(\hbar)$.

We want to show that the operator $S_\hbar$ being defined only on
$\mathcal{A}(\hbar)$ is also continuous in the $C^*$-norm and thus
extends to $\mathfrak{A}(\hbar)$.  To show the continuity of
$S_{\hbar}$ we will need the following lemma that shows that there is
a star root of the Gauß function.
\begin{lemma}
    \label{lemma:RootGauss}
    Let $G_\hbar$ be the normalized Gauß function as in
    \eqref{eq:NormalizedGauss} used to define the operator
    $S_{\hbar}$. Then we have
    \begin{equation}
        \label{eq:RootOfGauss}
        G_\hbar \star G_\hbar
        =
        \frac{1}{(2\pi\hbar)^n} \frac{1}{\sqrt{\det G}} G_\hbar.        
    \end{equation}
\end{lemma}
\begin{proof}
    The proof is a straightforward and well-known computation, see
    e.g. \cite[Prop.~3.3.1]{fedosov:1996a}.
\end{proof}
From equation \eqref{eq:ShbarSeries} and the trivial fact that
$\sqrt{\det G} a_{L=0} = S_\hbar(a)$ we obtain the following
statement:
\begin{lemma}[Leading order of $S_{\hbar}(a^*\star_{\hbar} a)$]
    \label{lemma:LeadingOrder}
    For $a \in \mathcal{A}^\infty$ we have
    \begin{equation}
        S_{\hbar}(a^*\star_\hbar a)
        =
        \frac{1}{\det G} S_{\hbar}(a^*)S_{\hbar}(a) + b,
    \end{equation}
    where $b \in \mathcal A^+$ is positive.
\end{lemma}
\begin{theorem}
    \label{theorem:continuity}
    Let $(\mathfrak{A}, \cdot, \norm{\cdot})$ be a $C^*$-algebra with
    isometric and strongly continuous action $\alpha$ of $V$ and let
    $\mathcal{A}(\hbar)=(\mathfrak{A}^\infty, \star_\hbar,
    \norm{\cdot}_\hbar)$ be the Rieffel deformed pre-$C^*$-algebra.
    Then the operator
    \begin{equation}
        \label{eq:ShbarContinuous}
        S_{\hbar}: \mathcal{A}(\hbar)\longrightarrow \mathfrak{A}
    \end{equation}
    is a continuous operator in the $C^*$-norms of $\mathcal{A}(\hbar)$
    and $\mathfrak{A}$.
\end{theorem}
\begin{proof}
    Since $\mathfrak{A}$ is a $C^*$-algebra, we have $\norm{S_{\hbar}
      a}^2 = \norm{(S_\hbar a)^*(S_\hbar a)}$. From
    Lemma~\ref{lemma:LeadingOrder} it follows that $(S_\hbar a)^*(S_\hbar
    a) \le \det(G) S_\hbar(a^* \star_\hbar a)$ in the sense of
    positive elements in $\mathfrak{A}$. From this it follows that the
    same holds for the norms, i.e.\ $\norm{(S_\hbar a)^*(S_\hbar a)}
    \le \det(G) \norm{S_\hbar(a^* \star_\hbar a)}$. In order to
    compute the last norm we need the following fact that
    \begin{equation}
    \label{eq:starkgeschlossen}
    \int_V f \star_\hbar g = \int_V fg
    \end{equation}
    for all $f, g \in \mathcal{S}(V, \mathfrak{A})$, see \cite[Lemma
    3.8]{rieffel:1993a}. Moreover, thanks to the fast decay of
    functions in $\mathcal{S}(V, \mathfrak{A})$, Equation
    \eqref{eq:starkgeschlossen} still holds if one of them is in
    $C^\infty_u(V, \mathfrak{A})$.  Using this and
    Lemma~\ref{lemma:RootGauss} we find
    \begin{align*}
        \norm{S_\hbar(a^* \star_\hbar a)}
        &=
        \det(G)
        \norm{
          \int_V (G_\hbar \star_\hbar \alpha(a^* \star_\hbar a)) 
          (u) \D u
        } \\
        &=
        (2\pi\hbar)^n (\det(G))^{\frac{3}{2}}
        \norm{\int_V
          \left(
              G_\hbar \star_\hbar G_\hbar \star_\hbar \alpha(a^*
              \star_\hbar a)
          \right) (u) \D u
        } \\
        &=
        (2\pi\hbar)^n (\det(G))^{\frac{3}{2}}
        \norm{\int_V
          \left(
              G_\hbar \star_\hbar \alpha(a)^* \star_\hbar \alpha(a)
              \star_\hbar G_\hbar
          \right) (u) \D u
        } \\
        &=
        (2\pi\hbar)^n (\det(G))^{\frac{3}{2}}
        \norm{
          \SP{
            \alpha(a) \star_\hbar G_\hbar,
            \alpha(a) \star_\hbar G_\hbar
          }
        } \\
        &\le
        (2\pi\hbar)^n (\det(G))^{\frac{3}{2}}
        \norm{G_\hbar}^2_{\mathcal{S}}
        \norm{a}^2_\hbar,
    \end{align*}
    by observing that $G_\hbar$ is central for the \emph{undeformed}
    pointwise product of $C^\infty_u(V, \mathfrak{A})$. Thus we have
    the desired continuity
    \begin{equation}
        \label{eq:OperatorNormShbar}
        \norm{S_\hbar a}^2 \le
        (2\pi\hbar)^n (\det(G))^{\frac{3}{2}}
        \norm{G_\hbar}^2_{\mathcal{S}}
        \norm{a}^2_\hbar.
    \end{equation}
\end{proof}
\begin{corollary}
    \label{corollary:omegaToomegahbar}
    Let $\omega: \mathfrak{A} \longrightarrow \mathbb{C}$ be a
    positive linear functional of the undeformed $C^*$-algebra. Then
    $\omega_\hbar = \omega \circ S_\hbar: \mathcal{A}(\hbar)
    \longrightarrow \mathbb{C}$ is continuous with respect to
    $\norm{\cdot}_\hbar$ and extends to a positive linear functional
    $\omega_\hbar: \mathfrak{A}(\hbar) \longrightarrow \mathbb{C}$.
\end{corollary}
Thus we have constructed for every classical state $\omega$ a
corresponding quantum state using the operator $S_\hbar$.  We shall
also use the symbol 
\begin{equation}
    \label{eq:ShbarExtension}
    S_\hbar: 
    \mathfrak{A}(\hbar) \longrightarrow \mathfrak{A}
\end{equation}
for the extension of the operator $S_\hbar$ to the completions in the
corresponding $C^*$-topologies.

%
%

\section{Continuous Fields of States}
\label{sec:FieldOfStates}

In a last step we want to discuss in which sense $\omega_\hbar$ can be
considered as a deformation of $\omega$: clearly we have
$\omega (a) = \lim_{\hbar\searrow 0 }\omega_\hbar(a)$ pointwise for
every $a \in \mathcal{A}^\infty$ but we want to
show some continuity properties beyond that trivial observation.

One of the main results in Rieffel's work \cite{rieffel:1993a} is that
the deformed $C^*$-algebras $\{\mathfrak{A}(\hbar)\}_{\hbar \ge 0}$
actually yield a \emph{continuous field} in the sense of Dixmier
\cite{dixmier:1977b}: Recall that a continuous field
structure on a collection $\{\mathfrak{A}(\hbar)\}_{\hbar \ge 0}$ of
$C^*$-algebras consists in the choice of \emph{continuous sections}
$\Gamma \subseteq \prod_{\hbar \ge 0} \mathfrak{A}(\hbar)$ subject to
the following technical conditions: $\Gamma$ is a $^*$-algebra with
respect to the pointwise product of the sections and for each fixed
$\hbar$ the set of possible values $\{a(\hbar)\}_{a \in \Gamma}
\subseteq \mathfrak{A}(\hbar)$ is dense. For unital $C^*$-algebras, we
require that the unit section $\hbar \mapsto \Unit(\hbar) =
\Unit_{\mathfrak{A}(\hbar)}$ is always in $\Gamma$. Moreover, the
function $\hbar \mapsto \norm{a(\hbar)}_\hbar$ is continuous for all
$a \in \Gamma$.  Finally, if an arbitrary section $b \in \prod_{\hbar
  \ge 0} \mathfrak{A}(\hbar)$ can locally be approximated uniformly by
continuous sections, it is already continuous itself, i.e.\ if $b$ is a
section such that for all $\varepsilon > 0$ and all $\hbar_0$ there
exists an open neighborhood $U \subseteq [0, \infty)$ of $\hbar_0$ and
a continuous section $a \in \Gamma$ such that $\norm{a(\hbar) -
  b(\hbar)}_\hbar \le \varepsilon$ uniformly for all $\hbar \in U$,
then $b \in \Gamma$. It follows that $\Gamma$ necessarily contains
$C^0(\mathbb{R}^+_0)$.

In the case of the Rieffel deformation the $^*$-algebra of continuous
sections $\Gamma$ can be obtained from the ``constant'' sections
$a(\hbar) = a \in \mathfrak{A}^\infty$. In detail, one has the
following (technical) characterization:
\begin{proposition}
    \label{proposition:GammaNullGeneratesGamma}
    Let $\mathcal{A}(\hbar) = (\mathfrak{A}^\infty, \star_{\hbar},
    \norm{\cdot}_\hbar)$ be the Rieffel deformed pre-$C^*$-algebras
    and let $\{\mathfrak{A}(\hbar)\}_{\hbar \ge 0}$ be the
    corresponding field of $C^*$-algebras. Then let
    \begin{equation}
        \label{eq:GammaDef}
        \begin{split}
            \Gamma = 
            \Big\{ 
            b \in \prod_{\hbar \ge 0}\mathfrak{A}(\hbar)
            \; &\Big| \;
            \forall \varepsilon>0
            \forall \hbar_0 \ge 0
            \exists U(\hbar_0)
            \exists a \in \Gamma_0
            \forall \hbar\in U(\hbar_0): \\
            &\quad \norm{b(\hbar)-a(\hbar)}_{\hbar} \leq \varepsilon
            \Big\}
        \end{split}
    \end{equation}
    be the set of sections generated by the set $\Gamma_0$ of
    sections. Then for all three choices
    \begin{enumerate}
    \item $\Gamma_0 = \mathfrak{A}^\infty$
    \item $\Gamma_0 = C^0(\mathbb R_0^+) \otimes \mathfrak{A}^\infty$
    \item $\Gamma_0$ is the $^*$-algebra generated by the vector space
        $C^0(\mathbb R_0^+) \otimes \mathfrak{A}^\infty$ with respect
        to $\star_\hbar$.
    \end{enumerate}
    the set $\Gamma$ is the same and defines the structure of a
    continuous field.
\end{proposition}
In other words, the $^*$-algebra $\Gamma$ of continuous sections yields
the smallest continuous field built on the collection
$\{\mathfrak{A}(\hbar)\}_{\hbar \ge 0}$ which contains the constant
sections $a: \hbar \mapsto a(\hbar) = a \in \mathfrak{A}^\infty$. The
second choice of $\Gamma_0$ is the smallest
$C^0(\mathbb{R}^+_0)$-module while the last choice corresponds to the
smallest $^*$-algebra containing $\mathfrak{A}^\infty$ and
$C^0(\mathbb{R}^+_0)$. In the following we shall always refer to this
continuous field structure $\Gamma$.

Turning back to the states we want to show that the set of states
$\omega_\hbar = \omega \circ S_\hbar$, where $\omega: \mathfrak{A}
\longrightarrow \mathbb{C}$ is a classical state, form a continuous
field of states in the following sense, see e.g.\
\cite[Def.~1.3.1]{landsman:1998a}:
\begin{definition}[Continuous field of states]
    \label{definition:ContinuousFieldStates}
    A continuous field of states on a continuous field of
    $C^*$-algebras $( \{\mathfrak A(\hbar)\}_{\hbar \ge 0}, \Gamma)$
    is a family of states $\omega_{\hbar}$ on $\mathfrak{A}(\hbar)$
    such that
    \begin{equation}
        \label{eq:ContinuousFieldStates}
        \hbar \mapsto \omega_\hbar (a(\hbar))
    \end{equation}
    is continuous for every continuous section $a \in \Gamma$.
\end{definition}
\begin{lemma}
    \label{lemma:ContinuityOfShbarahbar}
    If $a \in \Gamma$ is a continuous section, then the map
    $\mathbb{R}_0^+ \ni \hbar \mapsto S_{\hbar} a(\hbar) \in
    \mathfrak{A}$ is continuous in the (undeformed) $C^*$-norm of
    $\mathfrak{A}$.
\end{lemma}
\begin{proof}
    Note that here we use the extension of $S_\hbar$ to the completion
    $\mathfrak{A}(\hbar)$. Moreover, by
    Proposition~\ref{proposition:GammaNullGeneratesGamma} we can
    approximate $a$ by sections in $\Gamma_0 = C^0(\mathbb{R}^+_0)
    \otimes \mathfrak{A}^\infty$.  First, we show the continuity at
    $\hbar \ne 0$.
    \begin{align*}
	\norm{S_{\hbar}a(\hbar) - S_{\hbar'}a(\hbar')}
        &\leq
        \norm{S_{\hbar}a(\hbar) - S_{\hbar}a(\hbar')}
        +
        \norm{S_{\hbar}a(\hbar') - S_{\hbar'}a(\hbar')} \\
	&=
        \norm{S_{\hbar}(a(\hbar) - a_{\Delta\hbar}(\hbar))}
        +
        \norm{(S_{\hbar} - S_{\hbar'})(a(\hbar'))} \\
        &\leq
        c(\hbar) \norm{a(\hbar) - a_{\Delta\hbar}(\hbar)}_{\hbar}
        + \norm{(S_{\hbar}-S_{\hbar'})(a(\hbar'))},
    \end{align*}
    where $a_{\Delta\hbar}(\hbar) = a(\hbar + \Delta\hbar)$ with
    $\Delta\hbar = \hbar' - \hbar$ and $c(\hbar)$ is the constant from
    the estimate \eqref{eq:OperatorNormShbar}. It is now easy to see
    that the section $a_{\Delta\hbar}$ is approximated by sections of
    the form $\sum_n \tau_{\Delta \hbar} f_n a_n$, where $(\tau_{\Delta
      \hbar} f_n)(\hbar) = f_n(\hbar + \Delta\hbar)$. Thus $a_{\Delta
      \hbar}$ is still in $\Gamma$ and approximates $a$ for
    $\Delta\hbar \longrightarrow 0$. Hence the first term becomes
    small for $\hbar' \longrightarrow \hbar$. The second term requires
    more attention. We can approximate $a$ by sections of the form
    $\sum_n f_n a_n \in \Gamma_0$ with a finite sum and $f_n \in
    C^0(\mathbb{R}^+_0)$ and $a_n \in \mathfrak{A}^\infty$.  Then we
    have
    \begin{align*}
	&\norm{(S_{\hbar} - S_{\hbar'})(a(\hbar'))} \\
	&\leq 
        \norm{S_{\hbar} (a(\hbar') - \sum f_n(\hbar')a_n)} \\
        &\quad
        +
        \norm{(S_{\hbar} - S_{\hbar'})(\sum f_n(\hbar')a_n)}
        +
        \norm{S_{\hbar'}(a(\hbar') - \sum f_n(\hbar')a_n)} \\
	&\leq
        c(\hbar)
        \norm{
          a_{\Delta\hbar}(\hbar) 
          -
          \sum\tau_{\Delta\hbar}f_n(\hbar)a_n
        }_\hbar \\
        &\quad
        +
        \norm{\sum f_n(\hbar')a_n}
        \int\left|
            G_\hbar(u) - G_{\hbar'}(u)
        \right| \D u 
        + c(\hbar') \norm{a(\hbar') - \sum f_n(\hbar')a_n}_{\hbar'}.
    \end{align*}
    The constants $c(\hbar)$ and $c(\hbar')$ are bounded in a small
    neighborhood of $\hbar\neq 0$. Since the functions $f_n$ are
    continuous, $\|\sum f_n(\hbar')a_n\|$ is bounded on a
    neighborhood. The other factors become smaller than any
    $\varepsilon>0$ for $\hbar' \longrightarrow \hbar$. This shows the
    continuity at $\hbar \ne 0$. For the continuity at $0$ we have
    with $S_0 = \id$
    \[
    \norm{S_{\hbar}(a(\hbar)) - S_{0}(a(0))}
    \leq 
    \norm{S_{\hbar}(a(\hbar) - S_{\hbar}(a(0)))}
    + 
    \norm{S_{\hbar}(a(0)) - a(0)}.
    \]
    The first term gives
    \begin{align*}
        &\norm{
          \int G_\hbar(u) \alpha_u (a(\hbar)-a(0)) \D u 
        } \\ 
        &\quad\leq 
        \norm{a(\hbar) - a(0)} \int G_\hbar(u) \D u
        =
        \norm{a(\hbar) - a(0)},
    \end{align*}
    since the Gauß function is normalized and $\alpha$ is isometric.
    Now $a(\hbar) = a_{\hbar}(0)$ approximates $a(0)$ in a
    neighborhood of zero whence this contribution becomes small for
    $\hbar \searrow 0$.  The second term becomes small thanks to the
    asymptotics from Lemma~\ref{lemma:asymp} in the topology of
    $\mathfrak{A}$. This shows the continuity at $0$, too.
\end{proof}
From this lemma we obtain the main result immediately:
\begin{theorem}
    For every classical state $\omega: \mathfrak{A} \longrightarrow
    \mathbb{C}$ and for every continuous section $a \in \Gamma$ the map
    \begin{equation}
        \label{eq:OmegahbarContinuous}
        \hbar \mapsto \omega(S_{\hbar}(a(\hbar))) = \omega_{\hbar}(a(\hbar))
    \end{equation}
    is continuous. Hence $\{\omega_{\hbar}\}_{\hbar \ge 0}$ is a
    continuous field of states with $\omega_0 = \omega$.
\end{theorem}
\begin{remark}[Completely positive deformation]
    \label{remark:CompletePositive}
    Since with $\mathfrak{A}$ also the matrices $M_n(\mathfrak{A})$
    carry an induced action of $V$, we can repeat the whole
    deformation process for $M_n(\mathfrak{A})$. Then it is easy to
    see that the deformations of $(M_n(\mathfrak{A}))(\hbar)$ are just
    $M_n(\mathfrak{A}(\hbar))$. Thus the above statement on the
    deformation of states applies to $M_n(\mathfrak{A})$, too. In
    \cite{bursztyn.waldmann:2005a}, such deformations were called
    \emph{completely positive deformations}. Of course, here we obtain
    this statement in a strict framework and not for formal power
    series in $\hbar$.
\end{remark}

%
%

\begin{footnotesize}
    \renewcommand{\arraystretch}{0.5}
    
\end{footnotesize}

\end{document}